\documentclass{article}
\usepackage[utf8]{inputenc}

\title{Mean estimation in the add-remove model of differential privacy}

\usepackage{authblk} 
\author{Alex Kulesza\thanks{kulesza@google.com}}
\author{Ananda Theertha Suresh\thanks{theertha@google.com}}
\author{Yuyan Wang
\thanks{wangyy@google.com}}
\affil{Google Research, New York}

\usepackage{natbib}
\usepackage{graphicx}
\usepackage{fullpage}
\usepackage{hyperref}
\hypersetup{colorlinks,linkcolor={red},citecolor={blue},urlcolor={red}}
\usepackage{amsmath, amssymb, amsthm, algorithm,pseudo}
\usepackage{subcaption}
\usepackage{xcolor}
\newcommand{\ignore}[1]{{}}
\newcommand{\EE}{\mathbb{E}}
\newcommand{\cD}{\mathcal{D}}
\newcommand{\cA}{\mathcal{A}}

\renewcommand{\Pr}{\text{Pr}}
\newcommand{\clip}{\text{Clip}}

\newcommand{\lap}{\text{Lap}}
\newtheorem{theorem}{Theorem}
\newtheorem{lemma}{Lemma}
\newcommand{\ar}{\text{ar}}
\newcommand{\sw}{\text{sw}}
\newtheorem{definition}{Definition}
\newtheorem{corollary}{Corollary}

\newcommand{\conf}[1]{}
\newcommand{\arxiv}[1]{#1}
\begin{document}

\maketitle

\newcommand{\algname}{hourglass\,}
\newcommand{\capalgname}{Hourglass\,}
\renewcommand{\epsilon}{\varepsilon}

\begin{abstract}
Differential privacy is often studied under two different models of neighboring datasets: the \emph{add-remove} model and the \emph{swap} model. While the swap model is frequently used in the academic literature to simplify analysis, many practical applications rely on the more conservative add-remove model, where obtaining tight results can be difficult. Here, we study the problem of one-dimensional mean estimation under the add-remove model. We propose a new algorithm and show that it is min-max optimal, achieving the best possible constant in the leading term of the mean squared error for all $\epsilon$, and that this constant is the same as the optimal algorithm under the swap model. These results show that the add-remove and swap models give nearly identical errors for mean estimation, even though the add-remove model cannot treat the size of the dataset as public information. We also demonstrate empirically that our proposed algorithm yields at least a factor of two improvement in mean squared error over algorithms frequently used in practice. One of our main technical contributions is a new \emph{\algname} mechanism, which might be of independent interest in other scenarios.
\end{abstract}

\section{Introduction}

Mean estimation is one of the simplest and most widely used techniques in statistics, and is often deployed as a subroutine for more complex analyses. However, the mean of a dataset can reveal private information, and so a variety of differentially private methods have been proposed to estimate it. These include private mean estimation with robustness \citep{liu2021robust}, mean estimation under statistical models \citep{kamath2020private}, and techniques with instance-specific guarantees \citep{huang2021instance, dick2023subset}. 

However, despite the ubiquitous nature of differentially private mean estimation, fundamental questions remain about the optimal error that can be achieved under commonly used variants of differential privacy. We begin with the general definition of a differentially private mechanism.

\begin{definition}[Differential privacy \citep{dwork2014algorithmic}]
A randomized, real-valued algorithm $A$ satisfies $\varepsilon$-differential privacy if for any two neighboring datasets $D, D'$ and for any output $\mathcal{S} \subseteq \mathcal{R}$, it holds that 
 \[
 \text{Pr}[A(D) \in \mathcal{S}] \leq e^{\varepsilon} \cdot \text{Pr}[A(D') \in \mathcal{S}].
 \]
\end{definition}
It remains to specify what makes two datasets \emph{neighbors}. Two definitions are commonly used. The first, called the \emph{swap} model \citep{dwork2006calibrating, vadhan2017complexity}, defines two datasets $D$ and $D'$ as neighboring if and only if 
\[
|D \setminus D'| = 1 \text{ and } |D' \setminus D| = 1.
\]
The second, called the \emph{add-remove} model \citep[Definition 2.4]{dwork2014algorithmic}, defines $D$ and $D'$ as neighboring if and only if 
\[
|D \setminus D'| + |D' \setminus D| = 1.
\]
Intuitively, under the swap model, $D'$ is obtained from $D$ by changing a value, while under the add-remove model, it is obtained by adding or removing a value.

While both neighborhood models of differential privacy are studied in the literature, the add-remove model is more frequently used for statistical queries in practice \citep{
mcsherry2009privacy, wilson2019differentially, rogers2020linkedin, amin2022plume}, likely because it is more conservative: the add-remove model protects the \emph{size} of the input dataset, while the swap model does not. Furthermore, a $\varepsilon$-DP algorithm under the add-remove model is also a $2\varepsilon$-DP algorithm under the swap model, so in this sense the relationship is strict.

Here, we revisit the problem of scalar mean estimation in the add-remove model of differential privacy, proposing a new algorithm that is min-max optimal, including the constant on the leading term of the error, and showing that the add-remove and swap models give nearly identical errors despite the former's additional protections.

\subsection{Setting}

Let $\cD_n(\ell, u)$ denote the set of all datasets consisting of $n$ real values in the range $[\ell, u]$, let $\cD_{\geq n}(\ell, u)$ denote the set of all datasets consisting of \emph{at least} $n$ real values in the range $[\ell, u]$, and let $\cD^*(\ell, u)$ denote the set of datasets of all sizes with values in $[\ell,u]$.

Given a dataset $D = \{x_1, x_2,\ldots,\}$, the mean is given by
\[
\mu(D) \triangleq \frac{1}{|D|} \sum_{x \in D} x.
\]
We measure the utility of a mean estimator $\hat{\mu}: \cD^*(\ell, u) \to [\ell, u]$ for a dataset $D$ in terms of mean squared error (MSE),
\[
L(\hat{\mu}, D) \triangleq \EE[(\hat{\mu}(D) - \mu(D))^2],
\]
where the expectation is over the randomization of the estimator. Typically, for private estimators, $L(\hat{\mu}, D)$ decreases with the size of $D$ as $O(1/|D|^2)$. Hence, we measure the normalized mean squared error as 
\[
L_{\text{norm}}(\hat{\mu}, D)  \triangleq |D|^2 L(\hat{\mu}, D).
\]

Let $\cA^{\sw}_\varepsilon$ denote the set of all $\varepsilon$-differentially private algorithms under the swap model, and let $\cA^{\ar}_\varepsilon$ denote the set of all $\varepsilon$-differentially private algorithms under the add-remove model. The min-max normalized mean squared error for sufficiently large datasets in the swap model is defined as 
\[
R_{\sw}(\varepsilon, n_0, \ell, u)  \triangleq \inf_{\hat{\mu} \in \cA^{\sw}_\varepsilon} \sup_{D \in \cD_{\geq n_0}(\ell, u)} L_{\text{norm}}(\hat{\mu}, D),
\]
and similarly in the add-remove model is
\[
R_{\ar}(\varepsilon, n_0, \ell, u)  \triangleq \inf_{\hat{\mu} \in \cA^{\ar}_\varepsilon} \sup_{D \in \cD_{\geq n_0}(\ell, u)} L_{\text{norm}}(\hat{\mu}, D).
\]

\subsection{Optimality in the swap model}

\citet{geng2014optimal} showed that
\begin{equation}
R_{\sw}(\varepsilon, n_0, \ell, u) = (u-\ell)^2 \cdot \sigma^2(\epsilon),
\label{eq:swap_all_epsilon}
\end{equation}
where 
\begin{equation}
\label{eq:opt_error}
\sigma^2(\epsilon) = \frac{2^{-2/3} e^{-2\epsilon/3}(1+e^{-\epsilon})^{2/3} + e^{-\epsilon}}{(1-e^{-\epsilon})^2}.
\end{equation}
As $\epsilon \to 0$, $\sigma^2(\epsilon) \to 2/\epsilon^2$, and hence for small values of $\epsilon$
\begin{equation}
R_{\sw}(\varepsilon, n_0, \ell, u) = \frac{2(u-\ell)^2}{\varepsilon^2} \left(1 \pm o(1) \right),
\label{eq:swap}
\end{equation}
where the $o(1)$ term tends to zero as $n_0 \to\infty$. The Laplace mechanism matches this mean squared error up to the $o(1)$ term, and in this sense is optimal as $\epsilon \to 0$.

However, for larger $\epsilon$, the Laplace mechanism is not optimal. \citet{geng2014optimal} defined a class of differentially private mechanisms called \emph{staircase} mechanisms whose density is parameterized by $\gamma \in [0, 1]$ and showed that, for any monotonic loss, there exists a $\gamma$ such that the staircase mechanism is the optimal differentially private mechanism for that loss. We provide a definition of the staircase mechanism in Definition~\ref{def:1d_stair} for completeness.

\subsection{Optimality in the add-remove model}

The story is a bit more complicated under the add-remove model. Since the mean is the ratio of the sum ($s \triangleq \sum_{x \in D} x$) to the count ($n \triangleq |D|$), one simple algorithm for private mean estimation is to use a fraction of the privacy budget (say, $\varepsilon/2$) to estimate the sum as $\hat{s}$, use the remaining privacy budget ($\varepsilon /2$) to estimate the count as $\hat{n}$, and finally estimate the mean as $\hat{s} / \hat{n}$. Since the true mean always lies in the range $[\ell, u]$, we can additionally clip the result to $[\ell, u]$ to improve accuracy. This standard algorithm is shown in Algorithm~\ref{alg:very_old}, where $\clip(x, [a, b]) = \max(a, \min(x, b))$.

\begin{minipage}{0.44\textwidth}
\begin{algorithm}[H]
\noindent\textbf{Input:} Multiset $D \subset [l,u]$, $\varepsilon > 0$.
\begin{pseudo}
    Let $w = \max(|\ell|, |u|)$.\\
    Let  $s = \sum_{x\in D}x$. \\
    Let $n = |D|$. \\
    Let $\hat s = s + Z_s$, where $Z_s \sim \lap(\frac{2w}{\varepsilon})$.\\
    Let $\hat n = n + Z_n$, where $Z_n \sim \lap(\frac{2}{\varepsilon})$.\\
    Output $\hat \mu = \clip(\frac{\hat s}{\hat n}, [\ell, u])$.
\end{pseudo}
\caption{Independent noise addition.}
\label{alg:very_old}
\end{algorithm}
\end{minipage}
\hfill
\begin{minipage}{0.44\textwidth}
\begin{algorithm}[H]
\noindent\textbf{Input:} Multiset $D \subset [l,u]$, $\varepsilon > 0$.
\begin{pseudo}
    Let $w = u - l$ and $m = \frac{l+u}{2}$.\\
    Let $D' = D - m$.\\
    Let  $s = \sum_{x\in D'}x$. \\
    Let $n = |D'|$. \\
    Let $\hat s = s + Z_s$, where $Z_s \sim \lap(\frac{w}{\varepsilon})$.\\
    Let $\hat n = n + Z_n$, where $Z_n \sim \lap(\frac{2}{\varepsilon})$.\\
    Output $\hat \mu = \clip(\frac{\hat s}{\hat n}, [-\frac{w}{2},\frac{w}{2}]) + m$.
\end{pseudo}
\caption{Shifted noise addition.}
\label{alg:current}
\end{algorithm}
\end{minipage}
\vspace{2ex}

The noise added in Algorithm~\ref{alg:very_old} is proportional to $\max(|\ell|, |u|)$, which can be badly suboptimal, for instance if $l = 10^6$ and $u = 10^6+1$. Algorithm~\ref{alg:current} modifies Algorithm~\ref{alg:very_old} by shifting the inputs by $(\ell+u)/2$ before computing the sum, reducing the sensitivity to $(u-\ell)/2$. This improves the error, sometimes dramatically. It can be shown that, for any dataset $D \in \cD^*(\ell, u)$,
\[
L_{\text{norm}}(\text{Algorithm~\ref{alg:current}}, \ell, u) \leq \frac{4(u-\ell)^2}{\varepsilon^2} \left(1 + o(1) \right)
\]
and hence
\begin{equation}
\label{eq:current}
R_{\ar}(\varepsilon, n_0, \ell, u) \leq \frac{4(u-\ell)^2}{\varepsilon^2} \left(1 + o(1) \right).
\end{equation}
As before, $o(1)$ indicates a term that vanishes as the minimum dataset size $n_0$ grows; for simplicity, we drop $n_0$ in the notation of $R_{\sw}$ and $R_{\ar}$ for the rest of the paper.

This result is still a factor of two larger than the swap model lower bound in~\eqref{eq:swap}, and the loss is even further from $R_{\sw}$ in the low-privacy regime where $\epsilon$ is large. Recently, \citet{kamath2023bias} proposed a generic mechanism to convert a differentially private algorithm in the swap model to a differentially private algorithm in the add-remove model and instantiated it for the unbiased mean estimation problem. However, their focus was not the constant in the mean squared error, and the stated result \citep[Theorem D.6]{kamath2023bias} has a constant of $9$.

Thus, a natural question remains: is mean estimation in the add-remove model inherently harder than in the swap model? We show in this paper that the answer is no.

\section{Our contributions}

We propose a new mean estimation algorithm in Algorithm~\ref{alg:new}, introducing two key improvements over Algorithm~\ref{alg:current}:

\textbf{Transformed noise addition.} Instead of estimating the sum and count directly as in Algorithm~\ref{alg:current}, the new algorithm estimates a linear transformation of the sum and count to reduce $\ell_1$ sensitivity. This is sufficient to achieve optimal error using the vector Laplace mechanism in the high-privacy regime where $\epsilon$ is small.

\textbf{The \algname mechanism.} 
To achieve optimal error in the low-privacy regime where $\epsilon$ is large, we propose a new two-dimensional noise distribution called the \emph{\algname} mechanism. It has the desirable property that the marginal distribution over either dimension is the optimal univariate staircase mechanism \citep{geng2014optimal}.

We show that the \algname mechanism can be sampled efficiently, and prove a bound on the mean squared error of Algorithm~\ref{alg:new} when the noise is drawn from the \algname mechanism. Combined with an information-theoretic lower bound on $R_{\ar}(\epsilon, \ell, u)$, this shows that Algorithm~\ref{alg:new} is optimal for all $\epsilon$, $\ell$, and $u$.

These bounds also match the result of \citet{geng2014optimal} for $R_{\sw}(\epsilon, \ell, u)$, establishing that the swap model and the add-remove model give the same mean squared error (up to $o(1)$ terms).

We note in passing that, when the bounds $\ell$ and $u$ are unknown, Algorithm~\ref{alg:new} can be combined with standard clipping algorithms to perform on unbounded domains \citep{amin2019bounding}.

\subsection{Overview of technical results}

We first analyze Algorithm~\ref{alg:current} as a baseline and provide a dataset specific upper bound in Lemma~\ref{lem:current}, proving that for any dataset $D$, its mean squared error is upper bounded by 
\begin{align*}
\left( \frac{2(u-l)^2}{|D|^2\varepsilon^2} + \frac{8(\mu-\frac{u+\ell}{2})^2}{|D|^2\varepsilon^2}\right) \left( 1 + o(1)\right).
\end{align*}

We then analyze Algorithm~\ref{alg:new} with Laplace noise in Theorem~\ref{thm:new}, showing that its error is at most 
\begin{align}
 \left( \frac{(u-l)^2}{|D|^2\varepsilon^2} + \frac{4(\mu-\frac{u+\ell}{2})^2}{|D|^2\varepsilon^2}\right) \left( 1 + o(1)\right),
 \label{eq:contrib_temp}
\end{align}
and hence
\[
R_{\text{ar}}(\varepsilon, \ell, u) \leq \frac{2(u-\ell)^2}{\varepsilon^2} \left(1 + o(1) \right).
\]
This is the same as the min-max MSE of the swap model for small values of $\epsilon$.


We next analyze Algorithm~\ref{alg:new} with noise drawn from the two-dimensional staircase mechanism proposed by \citet{geng2015staircase} in Lemma~\ref{lem:two_stair_mean}, and show that its error is at most
\begin{align*}
 \left( \frac{(u-l)^2 \widetilde{\sigma}^2(\epsilon)}{|D|^2}\right) \left( 1 + o(1)\right),
\end{align*}
where $\widetilde{\sigma}^2(\epsilon)$ is the variance of the two-dimensional staircase mechanism optimized for mean squared error with privacy guarantee $\epsilon$.
While the above result is better than~\eqref{eq:contrib_temp} for large values of $\epsilon$, it still does not match the error of the swap model in general. This is due to the fact that for large values of $\epsilon$, the error of the swap model scales as
\[
\sigma^2(\epsilon) =   \Theta \left( (u-\ell)^2 e^{-2\epsilon/3} \right),
\]
while the error of the two dimensional staircase mechanism applied in Algorithm~\ref{alg:new} scales as 
\[
\widetilde{\sigma}^2(\epsilon)  =  \Theta \left( (u-\ell)^2 e^{-\epsilon/2} \right).
\]
(See Lemma~\ref{lem:two_stair} in the Appendix for details.)


We finally analyze Algorithm~\ref{alg:new} with noise drawn from the \algname mechanism in Theorem~\ref{thm:hourglass_main}, showing that its error is upper bounded by 
\begin{align*}
 \left( \frac{(u-l)^2 {\sigma}^2(\epsilon)}{|D|^2} \right)\left( 1 + o(1)\right),
\end{align*}
where $\sigma^2(\epsilon)$ is given in~\eqref{eq:opt_error}, matching the swap model lower bound in~\eqref{eq:swap_all_epsilon}.

In addition, we prove an information-theoretic lower bound for the add-remove model in Theorem~\ref{thm:lower}:
\[
R_{\text{ar}}(\varepsilon, \ell, u) \geq (u-\ell)^2 \cdot \sigma^2(\epsilon) \cdot \left(1 - o(1) \right),
\]
establishing that 
\[
R_{\text{ar}}(\varepsilon, \ell, u) = (u-\ell)^2 \cdot \sigma^2(\epsilon) \cdot \left(1 \pm o(1) \right),
\]
and therefore that Algorithm~\ref{alg:new} is optimal, as well as that the add-remove and swap models have equivalent mean squared error for mean estimation.



\begin{algorithm}[t]
\noindent\textbf{Input:} Multiset $D \subset [l,u]$, $\varepsilon > 0$.
\begin{pseudo}
    Let $w = u - l$.\\
    Let $D' = D - l$.\\
    Let $s_1 = \sum_{x\in D'}x / w$. \\
    Let $s_2 = \sum_{x\in D'}(1-x/w)$. \\
    Let $Z =  [Z_1, Z_2] \sim \text{two-dim. noise mechanism}(\epsilon) $.     \\
    Let $\hat s_1 = s_1 + Z_1$.\\
    Let $\hat s_2 = s_2 + Z_2$.\\
    Output $\hat \mu = w \cdot \clip(\frac{ \hat s_1}{\hat s_1 + \hat s_2}, [0, 1]) + l$.
\end{pseudo}
\caption{Transformed noise addition.}
\label{alg:new}
\end{algorithm}

The rest of the paper is organized as follows. In Section~\ref{sec:high} we discuss the high-privacy regime, showing how a linear transformation on the sum and count leads to optimal error using the Laplace mechanism. In Section~\ref{sec:all}, we generalize our results result to the low-privacy regime, introducing the \algname mechanism in Section~\ref{sec:hour} and applying it to mean estimation in Section~\ref{sec:hour_mean}.  In Section~\ref{sec:lower} we prove an information-theoretic lower bound showing that our results with the \algname mechanism are optimal in the add-remove model and match the optimal error in the swap model as well. Finally, in Section~\ref{sec:experiments}, we empirically demonstrate the performance of our algorithm.

\begin{figure*}
    \centering
    \includegraphics[width=1.0\textwidth]{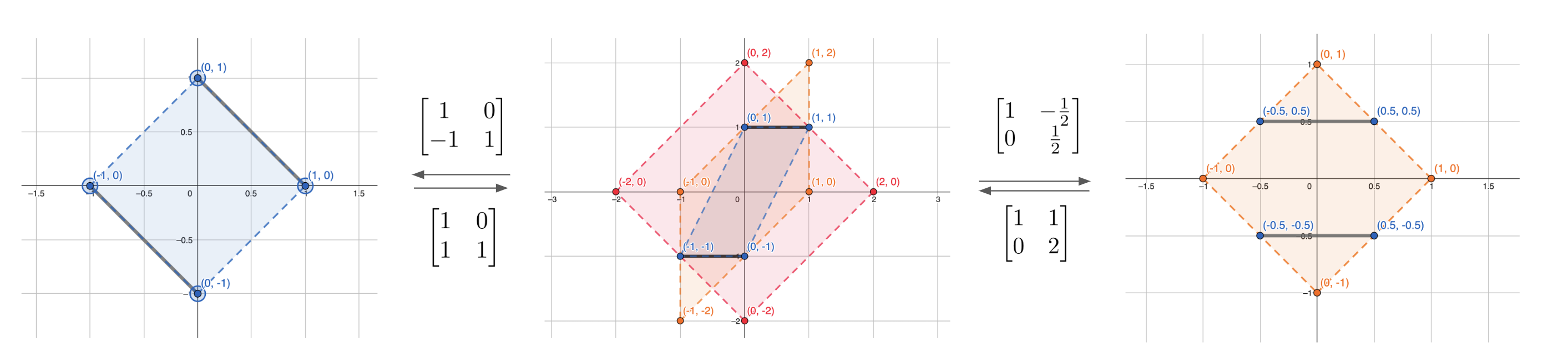}
    \caption{The Laplace mechanism applied to different linear transformations of $S(q)$. The middle plot shows the original sensitivity space, where \color{red}red \color{black}denotes the noise ball used by Algorithm~\ref{alg:very_old}; \color{orange}orange
    \color{black}the ball used by Algorithm~\ref{alg:current}, and  \color{blue}blue \color{black}the ball used by Algorithm~\ref{alg:new}, which is the smallest convex shape possible. The left plot shows the transformed space used by Algorithm~\ref{alg:new}, and the right plot shows the transformed space used by Algorithm~\ref{alg:current}.}
    \label{fig:unit_balls}
\end{figure*}

\section{High privacy regime}
\label{sec:high}

We first build intuition by viewing Algorithms~\ref{alg:very_old} and \ref{alg:current} geometrically, drawing on the framework of \citet{hardt2010geometry}. To simplify, we will assume $\ell = 0$ and $u = 1$. Let 
$$
q(D) = \left[\sum_{x \in D}x, \sum_{x \in D}1\right]
$$
be the two-dimensional sum and count vector for dataset $D$. Define the \emph{sensitivity space} $S(q)$ to be the set of possible values for $q(D) - q(D')$ when $D$ and $D'$ are neighboring datasets. Under the add-remove model, the sensitivity space for $q$ is $\pm[x, 1]$ for $x \in [0, 1]$, depicted by the two bold line segments in the middle plot of Figure~\ref{fig:unit_balls}. 

The standard vector Laplace mechanism can be used to obtain a differentially private estimate of $q(D)$ based on its sensitivity. In particular, the mechanism adds noise scaled to the maximum $\ell_1$ norm of the sensitivity space---that is, the smallest constant $a$ such that the $\ell_1$ ball $\{x: \|x\|_1 = a\}$ contains $S(q)$. Here, the minimum value of $a$ is 2, as shown by the red diamond in the middle plot of Figure \ref{fig:unit_balls}. And, indeed, adding Laplace noise scaled to $a=2$ is precisely what Algorithm~\ref{alg:very_old} does.

However, noise added in this way actually supports a much larger sensitivity space than just $S(q)$, as can be seen in the figure. This means that Algorithm~\ref{alg:very_old} is effectively ``wasting'' noise to protect against changes to $q$ that cannot occur, unnecessarily increasing error. (The problem is even worse when the range $[\ell, u]$ is far from zero, as noted earlier.)

Compared to this naive approach, Algorithm~\ref{alg:current} is significantly better, since the shift in Line 2 maps values in the dataset from $[0,1]$ to $[-\frac{1}{2}, \frac{1}{2}]$. The bold line segments in the right plot of Figure \ref{fig:unit_balls} depict the new sensitivity space, which is contained by the $\ell_1$ ball with smaller scale $a=1$ (shown in orange).

In fact, one can show that Algorithm~\ref{alg:current} is equivalent to the following procedure:
\begin{enumerate}
    \item  Apply a linear transformation given by the matrix $\begin{bmatrix} 1 & -\frac{1}{2} \\ 0 & \frac{1}{2} \end{bmatrix}$ to the sum and count vector $q(D)$.
    \item Add vector Laplace noise to the transformed vector according to its (reduced) sensitivity ($a=1$).
    \item Reverse the transformation by applying the inverse matrix $\begin{bmatrix} 1 & 1 \\ 0 & 2 \end{bmatrix}$, then divide to estimate the mean and truncate as before.
\end{enumerate}
In the original (untransformed) $q$-space, depicted in the middle plot of Figure \ref{fig:unit_balls}, the region protected by the resulting noise distribution is a parallelogram, shown in orange. Compared with the red diamond of Algorithm~\ref{alg:very_old}, this region more tightly encloses the sensitivity space, reducing noise and increasing accuracy. However, it still does not perfectly enclose $S(q)$.

\subsection{Optimizing the Laplace mechanism via linear transformation}
\label{sec:knorm_geometric}

The approach in Algorithm~\ref{alg:new} is to transform $q$ so that the $\ell_1$ unit ball encloses it as tightly as possible. To do this, it maps the two segments of the sensitivity space $[(0,1), (1,1)]$ and $[(-1,-1), (0, -1)]$ onto two sides of the $\ell_1$ ball using the transformation shown in the left plot of
Figure~\ref{fig:unit_balls}. The enclosing $\ell_1$ ball is depicted in blue.

More concretely, Algorithm~\ref{alg:new} computes the 
the scaled sum (denoted by $s_1$) and the difference of count and scaled sum (denoted by $s_2$). It then privatizes both $s_1$ and $s_2$ and uses them to compute mean. Next, we prove that Algorithm~\ref{alg:new} instantiated with Laplace noise is not just intuitively better but in fact optimal
for minimizing the mean squared error of the mean when $\epsilon$ is sufficiently small.

\subsection{Analysis of Laplace noise algorithms}
\label{sec:upper}
We first state a technical result which we use in proving upper bounds. 
\begin{lemma}
\label{lem:technical}
Let $b \geq 0 $. Let $a$ be such that $|a|/b \leq M$. Let $C = \frac{Z_a}{b} - \frac{aZ_b}{b^2}$ and $ F = \left \lvert \frac{2MZ^2_b}{b^2}\right \rvert + \left \lvert  \frac{2Z_a Z_b}{b^2} \right \rvert $. Then,
\begin{align*}
\EE \left[ \left(\clip\left( \frac{a + Z_a}{b+Z_b}  \right)- \frac{a}{b}   \right)^2\right]
& \leq \EE\left[C^2\right] + \EE \left[F^2\right] \\
& + 2 \sqrt{\EE \left[ C^2 \right] \EE\left[ F^2\right]}  \\
& + 4 M^2 \Pr(Z_b < - b/2).
\end{align*}
\end{lemma}
We provide the proof in Appendix~\ref{app:technical}. We next state an upper bound on the mean squared error of Algorithm~\ref{alg:current}.
\begin{lemma}
\label{lem:current}
The output of Algorithm~\ref{alg:current} is $\varepsilon$-differentially private. For any dataset $D \in \cD^*(\ell, u)$, 
the mean squared error of Algorithm~\ref{alg:current} is upper bounded by 
\begin{align*}
& \left( \frac{2(u-l)^2}{|D|^2\varepsilon^2} + \frac{8(\mu-m)^2}{|D|^2\varepsilon^2}\right) \left( 1 + o(1)\right) \\
& \leq   \frac{4(u-l)^2}{|D|^2\varepsilon^2}\left( 1 + o(1)\right),
\end{align*}
where $m = \frac{\ell + u}{2}$.
\end{lemma}
We provide the proof in Appendix~\ref{app:current}. We finally prove the upper bound on the mean squared error of Algorithm~\ref{alg:new} when the noise distribution is Laplace.
\begin{theorem}
Let $Z_1$ and $Z_2$ be sampled from independent Laplace distributions with parameter $\epsilon$. Then 
the output of Algorithm~\ref{alg:new} is $\varepsilon$-differentially private. Furthermore, for any dataset $D \in \cD^*(\ell, u)$, 
the mean squared error of Algorithm~\ref{alg:new} is upper bounded by \begin{align*}
& \left( \frac{(u-l)^2}{|D|^2\varepsilon^2} + \frac{4(\mu-m)^2}{|D|^2\varepsilon^2}\right) \left( 1 + o(1)\right) \\
& \leq  \frac{2(u-l)^2}{|D|^2\varepsilon^2} \left( 1 + o(1)\right),
\end{align*}
where $m = \frac{\ell+u}{2}$.
\label{thm:new}
\end{theorem}
\begin{proof}
Let $s = [s_1, s_2]$ be a two-dimensional vector. Let $s'$ and $s''$ be vectors corresponding to two neighboring datasets. Then the $\ell_1$ sensitivity is bounded by $\|s' - s''\|_1 \leq 1$.

In lines $5$ and $6$ of the algorithm, we add $\lap(1/\varepsilon)$ noise to each coordinate. Hence $[\hat{s}_1, \hat{s}_2]$ is an $\varepsilon$ differentially private vector and, by the post-processing lemma, the output $\hat{\mu}$ is $\varepsilon$ differentially private.

The proof of the error bound relies heavily on Lemma~\ref{lem:technical}. Let $n = |D|$. Since clipping only reduces the error,
\[
\EE[(\hat{\mu} - \mu)^2]  \leq (u - \ell)^2 \EE \left[ \left(\frac{ \hat s_1}{\hat s_1 + \hat s_2} - \frac{s_1}{n}  \right)^2\right].
\]
Let $a = s_1$, $b = n$, $Z_a = Z_{1} \sim \lap(1/\varepsilon)$, $Z_b = Z_1 + Z_2$, and $M = 1$. To usefully apply Lemma~\ref{lem:technical}, we need bounds for $\EE[C^2]$, $\EE[F^2]$, and $\Pr(Z_b < -b/2)$. We bound each of these below.
\begin{align}
\label{eq:expansion_mse}
\EE[C^2] & = \EE\left[\left(\frac{Z_1}{n} - \frac{(\frac{s_1}{n}) (Z_1 + Z_2)}{n} \right)^2 \right] \nonumber\\
 & =\frac{1}{n^2}\EE\left[\left(\left((1 - \frac{s_1}{n}\right)Z_1 - \frac{s_1}{n}Z_2 \right)^2 \right] \nonumber \\
 & = \frac{1}{n^2\varepsilon^2} + \frac{4(\frac{s_1}{n} - \frac{1}{2}))^2}{n^2\varepsilon^2}.
\end{align}
Since $(x+y)^2 \leq 2 x^2 + 2 y^2$,
\begin{align*}
\EE[F^2]  & \leq \frac{8}{n^4} \EE[(Z_1 + Z_2)^4] + \frac{8}{n^4} \EE[(Z_1 + Z_2)^2 Z^2_1] \\
& = o \left(\frac{1}{n^2\varepsilon^2} \right),
\end{align*}
where the last equality follows from the moments of the Laplace distribution \citep{kotz2012laplace}. Finally,
\begin{align*}
\Pr(Z_b < -b/2)  & \leq \Pr(Z_1 + Z_2 < -n/2) \\
& \leq \Pr(Z_1  < -n/4) + \Pr( Z_2 < -n/4) \\ 
&= o \left(\frac{1}{n^2\varepsilon^2} \right),
\end{align*}
where the last equality follows from the tail bounds of the Laplace distribution.
Combining the three bounds above with Lemma~\ref{lem:technical} and observing the fact that $\mu = \ell + \frac{s_1 w}{n}$ yields the theorem.
\end{proof}
Theorem~\ref{thm:new} implies the following corollary, which shows that, in the high-privacy regime where $\epsilon$ is small, the error of the add-remove model matches the swap model (Equation~\ref{eq:swap}).
\begin{corollary}
\[
R_{\ar}(\varepsilon, \ell, u) \leq \frac{2(u-l)^2}{\varepsilon^2}\left( 1 + o(1)\right).
\]
\end{corollary}

\section{Generalization to all values of $\epsilon$}
\label{sec:all}
In this section we design an optimal $\epsilon$-DP mechanism for private mean estimation in the add-remove model for \emph{any} $\epsilon$, dropping the high-privacy assumption. We prove the optimality of the new mechanism with respect to $R_{\ar}(\epsilon, \ell, u)$, and show that the optimal min-max error for the add-remove model is equivalent to that for the swap model, for any $\epsilon$, up to a $(1+o(1))$ constant factor.

We first motivate the new mechanism, which we call the \emph{\algname} mechanism. Observe that, in Algorithm \ref{alg:new}, $s_1 + s_2$ always sums to an integer, and hence the sensitivity space of $(s_1, s_2)$ is just the two bold segments on the left graph in Figure \ref{fig:unit_balls}. Previously, we used the Laplace mechanism to protect the convex hull of these segments, but this is (still) a strict superset of the actual sensitivity space, which is non-convex. Figure~\ref{fig:hourglass} shows how many points in the $\ell_1$ unit ball actually cannot be reached by taking a single step to a neighboring database. The \algname mechanism is designed to protect this sensitivity space more precisely, allowing for less noise. We formalize this notion and show that when the \algname mechanism is used in Algorithm \ref{alg:new} to noise $[s_1, s_2]$, the result is an optimal private mean estimator for all values of $\epsilon$.

\begin{figure}
\begin{subfigure}{0.46\textwidth}
    \centering
    \includegraphics[scale=0.12]{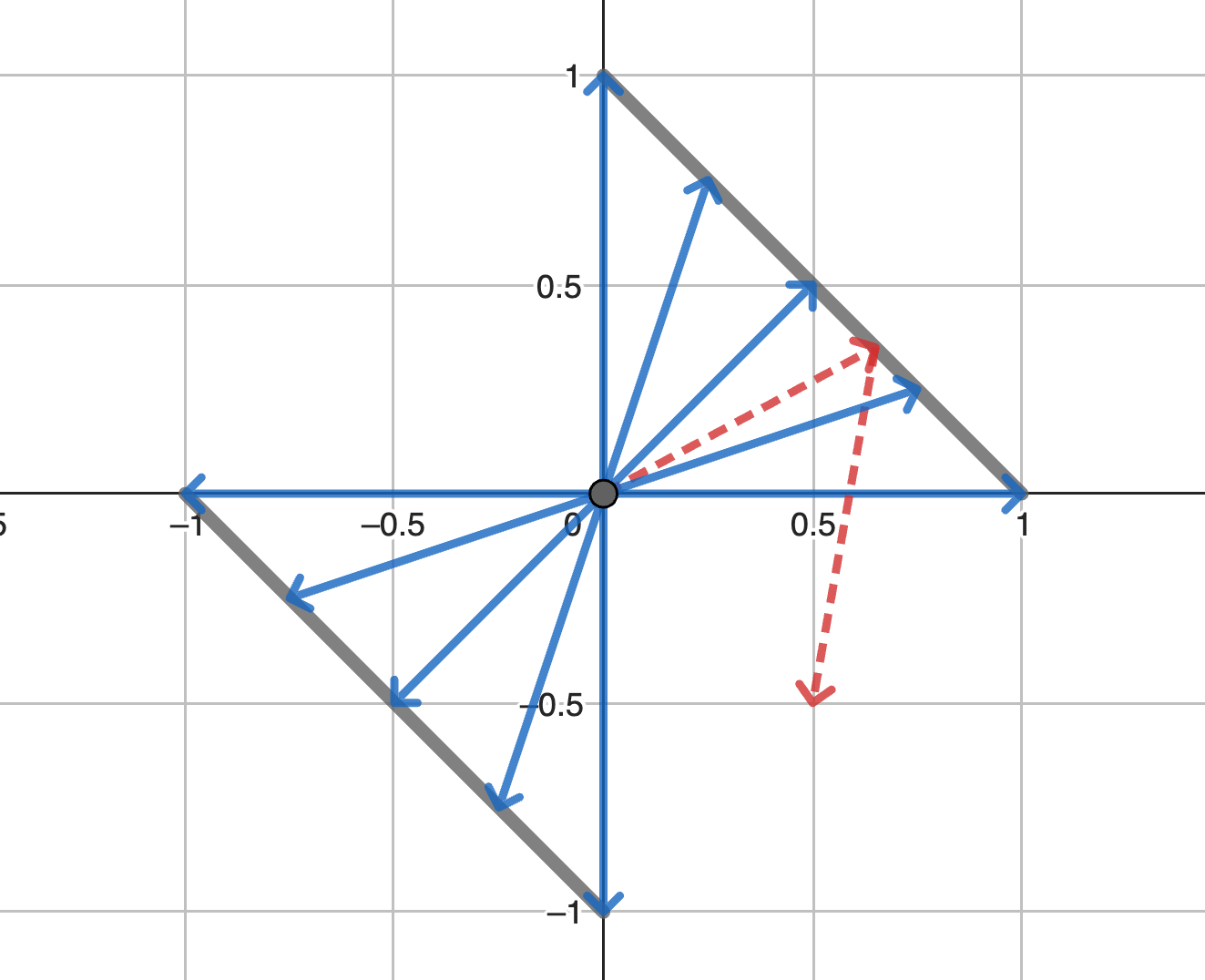}
    \caption{The ``hourglass'' sensitivity space. Blue solid arrows show how, when moving to a neighboring database, the origin can shift only to points along the line segments $(0,1)$ to $(1,0)$ and $(0,-1)$ to $(-1,0)$. Red dashed arrows show how points like $(0.5, -0.5)$ are only reachable via a chain of two neighboring steps, even though though they lie within the $\ell_1$ unit ball.}
    \label{fig:hourglass}
\end{subfigure}
\quad
\quad
\begin{subfigure}{0.46\textwidth}
    \centering
    \includegraphics[scale=0.62]{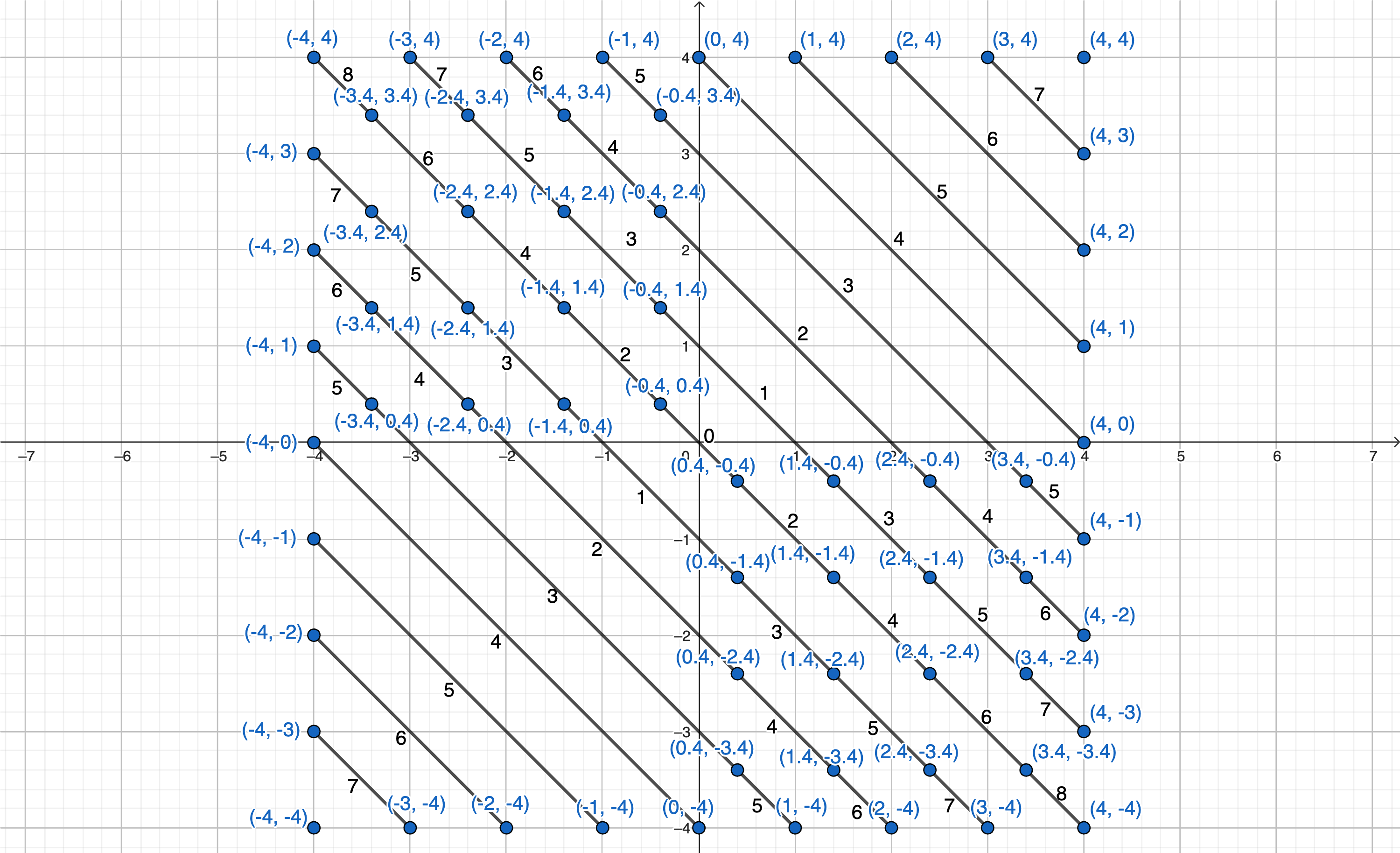}
    \caption{Illustration of the \algname distribution for $\Delta = 1$ and $\gamma=0.4$ in the range $[-4, 4]^2$. All points on a line segment marked $j$ have density proportional to $e^{-j\epsilon}$.}
    \label{fig:optimal_2d_staircase}
\end{subfigure}
\end{figure}

In particular, it is known that for privately computing one-dimensional sums (such as $s_1$ or $s_2$) the optimal mechanism is the staircase mechanism \cite{geng2014optimal}. The hourglass mechanism is constructed so that its marginal distributions both exactly match the univariate hourglass mechanism, with no change in $\epsilon$.

\subsection{\capalgname mechanism}
\label{sec:hour}

The \algname mechanism adds two-dimensional noise drawn from a distribution parameterized by $\gamma \in (0, 1]$. Its density $f_\gamma(x, y)$ is supported on $\{(x,y) : x+y= k\Delta, k \in \mathbb{Z}\}$.
For $x \geq 0$, we divide each diagonal line $x+y = k\Delta$ into regions according to the value of $x$. For integers $i\geq 1$:
\begin{align*}
A_k &: x \in [0, (k + \gamma)\Delta)\\
B_k(i) &: x \in [(k + \gamma + i - 1)\Delta, (k + \gamma + i)\Delta)
\end{align*}
Note that, for $k < 0$, $A_k$ is always empty, and $B_k(i)$ does not contain any points with $x \geq 0$ unless $i \geq -k$.

The density of the \algname noise distribution is given by
\begin{equation}
\label{eq:hourglass_prob}
f_\gamma(x,y) \propto
    \begin{cases}
     e^{-k \cdot \epsilon }, & (x, y) \in A_k\\
     e^{-(2i+k) \cdot \epsilon}, & (x, y) \in B_k(i)~.
    \end{cases}
\end{equation}
For $x < 0$, $f_\gamma(x,y) = f_\gamma(-x, -y)$. See Figure \ref{fig:optimal_2d_staircase} for an illustration of the \algname distribution when  $\Delta = 1$ and $\gamma = 0.4$.

Note that the \algname mechanism is different from the natural extension of the univariate staircase mechanism to two dimensions as proposed by \citet{geng2015staircase}. The latter is known to be optimal for an $\ell_1$-ball sensitivity space, but does not have the properties shown in Lemma \ref{lem:marginal_1d_staircase}, which are key to the optimality of the \algname mechanism, and does not perform as well in practice (see Section~\ref{sec:experiments}).

\begin{theorem}\label{thm:hourglass_dp}
Let $q: \cD \to R^2$ be a query such that for any two neighboring datasets $D$ and $D'$,
\[
q(D) - q(D') \in \{(x_0, \Delta-x_0) : x_0 \in [0,\Delta] \}.
\]
Then the \algname mechanism given by
$$q(D) + (Z_1, Z_2),$$ 
where $Z_1, Z_2$ are sampled according to the density $f_\gamma$, is $\epsilon$-DP.
\end{theorem}
\begin{proof}
We provide an intuition based on Figure~\ref{fig:optimal_2d_staircase} here and a detailed proof in Appendix~\ref{app:hour_privacy}. Recall that a segment marked by integer $j$ in Figure~\ref{fig:optimal_2d_staircase} has density proportional to $e^{-j\epsilon}$, and on neighboring databases the noise distribution is effectively shifted by $(x_0, 1-x_0)$ for some $x_0 \in [0,1]$. Thus, to ensure privacy, it must be the case that the integers marking any pair of points $(x,y)$ and $(x+x_0, y+1-x_0)$, which are on adjacent lines in Figure~\ref{fig:optimal_2d_staircase} and not more than one unit apart along either axis, differ by at most one. By construction, this is true across the support of $f_\gamma$.
\end{proof}

We now show that the marginal distributions of the \algname mechanism are univariate staircase mechanisms, which is the key ingredient to proving optimal utility guarantees.
\begin{lemma}[Marginal distribution]
\label{lem:marginal_1d_staircase}
The marginal densities $f_\gamma(x)$ and $f_\gamma(y)$ match the $\epsilon$-DP univariate staircase mechanism with parameter $\gamma$. That is, $f_\gamma(x) \propto e^{-\epsilon \lceil |x| - \gamma \rceil}$, and $f_\gamma(y)\propto e^{-\epsilon \lceil |y| - \gamma \rceil}$. \label{lem:hour_marginal}
\end{lemma}
\begin{proof}
Assume $x \geq 0$, and let $j \geq 0$ be the unique integer for which $x \in [j - 1 + \gamma, j + \gamma)$. Then the marginal density for $x$ is
\begin{align}
f_\gamma(x)
&\propto \sum_{k=-\infty}^{\infty} f_\gamma(x, k-x)\\
&= \sum_{k=-\infty}^{j-1} e^{-(2j-k)\epsilon}
+ \sum_{k=j}^{\infty} e^{-k\epsilon}
\end{align}
The first term follows because $j-k \geq 1$ in this range, and so $x \in B_k(j-k)$. The second term follows because $x \in A_k$ in this range. Summing the geometric series, the above is equal to
$$
e^{-(j+1)\epsilon}/(1-e^{-\epsilon})
+ e^{-j\epsilon}/(1-e^{-\epsilon})
\propto e^{-j\epsilon}~,
$$
which is precisely the staircase density. The densities for $x < 0$ and $y$ are handled symmetrically.
\end{proof}

Next we derive the conditional distribution of $y$ given $x$.
\begin{lemma}[Conditional distribution]
\label{lem:hrglass_conditional}
The conditional probability under $f_\gamma$ of $y$ given a fixed $x$ is a geometric distribution with ratio $e^{-\epsilon}$. In particular, for any $y$ such that $y+x$ is an integer,
$$f_\gamma(Y=y|x) = \frac{1-e^{-\epsilon}}{1+e^{-\epsilon}}e^{-\epsilon |y - y_0(x)|},$$
where
\begin{equation}
    y_0(x) =
    \begin{cases} 
    - x + \lfloor x + (1 - \gamma) \rfloor, & x \geq 0 \\
    -x - \lfloor - x + (1 - \gamma) \rfloor, & x < 0
    \end{cases}
\end{equation}
\end{lemma}
Using the above lemmas, we get a simple sampling algorithm for the \algname distribution: first sample $x$ from a univariate staircase mechanism with $\gamma$ and $\epsilon$, and then sample $y$ from the geometric distribution in Lemma \ref{lem:hrglass_conditional}.

\subsection{Implications for mean estimation}
\label{sec:hour_mean}

Before proceeding to analyze the accuracy of the hourglass mechanism for mean estimation, we first consider the two-dimensional staircase mechanism.
\begin{lemma}
Let $Z_1$ and $Z_2$ be sampled by the two-dimensional staircase mechanism with parameter $\epsilon$ and with parameter $\gamma$ optimized for mean squared error. Then the output of Algorithm~\ref{alg:new} is $\varepsilon$-DP. Furthermore, for any dataset $D \in \cD^*(\ell, u)$, 
the MSE of Algorithm~\ref{alg:new} is upper bounded by \begin{align*}
\left( \frac{(u-l)^2 \widetilde{\sigma}^2(\epsilon)}{|D|^2}\right) \left( 1 + o(1)\right),
\end{align*}
where $\widetilde{\sigma}^2(\epsilon)$ is the variance of the two-dimensional staircase mechanism optimized for MSE.
\label{lem:two_stair_mean}
\end{lemma}
The proof is provided in Appendix~\ref{app:two_stair_mean}. We now state our main upper bound for all values of $\epsilon$. The proof is similar to that of Lemma~\ref{lem:two_stair_mean} together with the fact that the marginal distribution along each dimension is the same as the univariate staircase mechanism (Lemma~\ref{lem:hour_marginal}). We provide the full proof in Appendix~\ref{app:hourglass_main}.

\begin{theorem}
Let $Z_1$ and $Z_2$ be sampled by the \algname mechanism with parameter $\epsilon$ and $\gamma$ as given by \citet[Equation 50]{geng2014optimal}. Then
the output of Algorithm~\ref{alg:new} is $\varepsilon$-differentially private. Furthermore, for any dataset $D \in \cD^*(\ell, u)$, 
the mean squared error of Algorithm~\ref{alg:new} is upper bounded by \begin{align*}
\left( \frac{(u-l)^2 {\sigma}^2(\epsilon)}{|D|^2}\right) \left( 1 + o(1)\right),
\end{align*}
where $\sigma^2(\epsilon)$ is defined in~\eqref{eq:opt_error}.
\label{thm:hourglass_main}
\end{theorem}

\begin{figure*}[h]
    \centering
    \begin{subfigure}[t]{0.3\textwidth}
        \centering
        \includegraphics[scale=0.37]{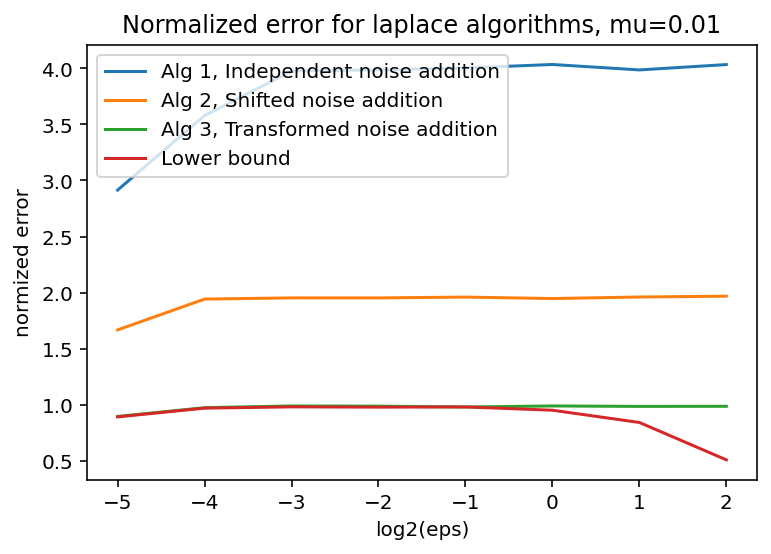}
        \caption{All Laplace mechanisms, high privacy regime, varying $\varepsilon$.}
          \label{fig:high_privacy_eps_vary}
    \end{subfigure}%
    ~~
    \begin{subfigure}[t]{0.3\textwidth}
        \centering
        \includegraphics[scale=0.37]{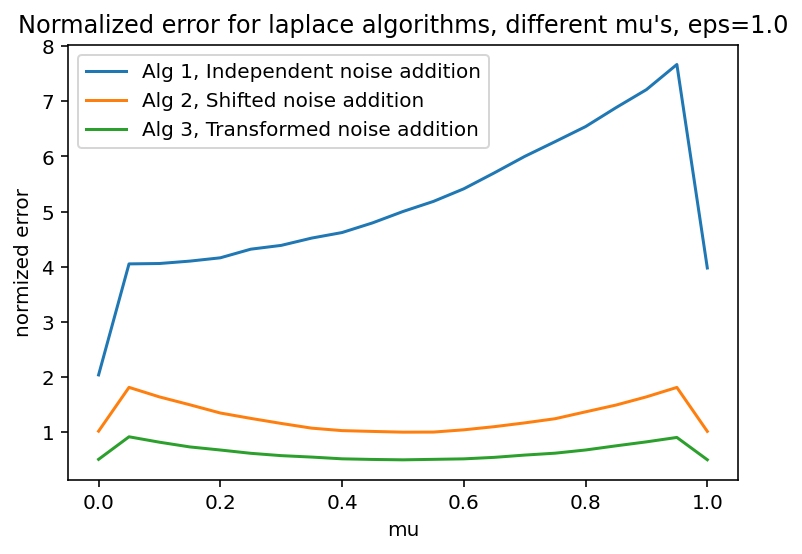}
        \caption{All Laplace mechanisms, high privacy regime, varying $\mu$.}
            \label{fig:high_privacy_mu_vary}
    \end{subfigure}
    ~~
    \begin{subfigure}[t]{0.3\textwidth}
        \centering
        \includegraphics[scale=0.37]{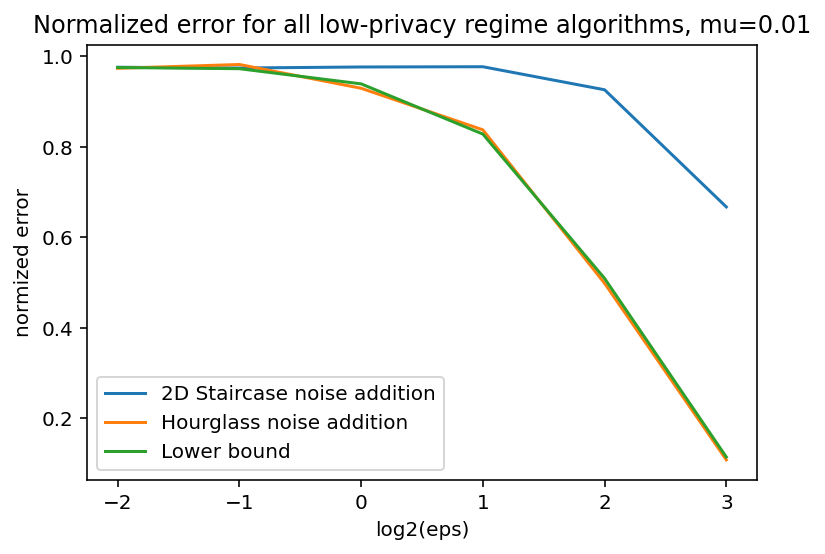}
        \caption{\capalgname and 2D staircase mechanisms, low privacy regime, varying $\varepsilon$.}
           \label{fig:low_privacy_eps_vary}
    \end{subfigure}
    \caption{Error comparison of different algorithms with varying $\varepsilon$ or $\mu$.}
    \label{fig:exps}
\end{figure*}

\section{Lower bound}
\label{sec:lower}

We next state a well-known result of \citet{geng2014optimal} that we will use to prove our lower bound.
\begin{lemma}[{\citet[Section VI.C]{geng2014optimal}}]
\label{lem:conjecture}
Let $k> 0$ and let $\cD^k$ denote the collection of all datasets of size $k$ where each value lies in the range $[0, 1]$. For $D \in \cD^k$, let $S(D)$ denote the sum of the elements in $D$. Let $\hat{S}$ denote an estimator of the same sum. Then, for any $\epsilon$,
\[
\inf_{\hat{S} \in \cA^{\ar}_{\epsilon}} \sup_{D \in \cD^k}  \EE \left[(\hat{S} - S(D))^2\right] \geq \sigma^2(\epsilon)(1- o(1)),
\]
where $\sigma^2(\epsilon)$ is given by~\eqref{eq:opt_error}. Here the $o(1)$ term goes to zero as $k$ tends to infinity.
\end{lemma}
Lemma~\ref{lem:conjecture} provides a lower bound for estimating the sum over datasets of fixed size. We construct a class of datasets $\cD'$ with varying size such that for every dataset in $\cD'$ there exists a corresponding dataset in $\cD^k$ (as defined in Lemma~\ref{lem:conjecture}) with the same sum. Hence, any differentially private mean estimator for the datasets in $\cD'$ can be modified to obtain a differentially private sum estimator for the datasets in $\cD^k$. We use this observation to show the following lower bound on the min-max MSE of mean estimation under the add-remove model of differential privacy.
\begin{theorem}
\[
R_{\ar}(\varepsilon, \ell, u) \geq  (u - \ell)^2 \cdot \sigma^2(\epsilon)(1- o(1)).
\]
\label{thm:lower}
\end{theorem}
Combining Theorem~\ref{thm:hourglass_main} together with Theorem~\ref{thm:lower} and~\eqref{eq:opt_error} yields the following result.
\begin{corollary}
\label{cor:hourglass_main}
For any $\epsilon > 0$, $R_{\ar}(\epsilon,  \ell, u)$ is equal to
$$ \frac{2^{-2/3} e^{-2\epsilon/3}(1+e^{-\epsilon})^{2/3} + e^{-\epsilon}}{(1-e^{-\epsilon})^2} \left( 1 \pm o(1) \right).$$
Furthermore, the \algname mechanism with suitable value of $\gamma$ achieves this normalized MSE.
\end{corollary}

\section{Experiments}
\label{sec:experiments}

In Figure \ref{fig:exps} we plot the empirical performance of the algorithms discussed in Sections~\ref{sec:high} and \ref{sec:hour} on synthetic datasets and explore how the performance changes with parameters such as the privacy budget $\varepsilon$ and the true mean $\mu$. The underlying datasets are generated i.i.d. with varying $\mu$ in the range $[\ell =0, u = 1]$. All datasets have $10,000$ points, and mean squared error is computed over $100,000$ runs of each algorithm. The errors are normalized by $\frac{|D|^2 \epsilon^2}{2}$ to keep them in a similar range across a wide range of parameters.

Figures~\ref{fig:high_privacy_eps_vary} and \ref{fig:high_privacy_mu_vary} compare the performance of Algorithms \ref{alg:very_old}, \ref{alg:current} and \ref{alg:new} using the Laplace mechanism, varying $\varepsilon$ and $\mu$, respectively. These plots focus on the high-privacy regime, where Algorithm \ref{alg:new} was shown to be optimal. Indeed, Algorithm~\ref{alg:new} outperforms the others, reducing error over Algorithms~\ref{alg:current} by roughly a factor of two, which matches the ratio of upper bounds in Lemma~\ref{lem:current} and Theorem~\ref{thm:new}. Similarly, in Figure~\ref{fig:high_privacy_eps_vary}, Algorithm~\ref{alg:new} closely matches the lower bound $\sigma^2(\epsilon)$ from \eqref{eq:opt_error} when $\epsilon < 1$.

In Figure~\ref{fig:high_privacy_mu_vary}, we explore how the error changes with the true mean of the data. As analyzed in \eqref{eq:expansion_mse}, the error is largest when $\mu$ approaches $0$ or $1$. However, the error also drops when $\mu$ becomes very close to $0$ or $1$, since $\frac{\hat{s_1}}{\hat{n}}$ falls out of the range $[\ell, u]$ with probability approaching 50\%. In these cases the clipping operation significantly reduces the mean squared error.

Finally, we compare the two-dimensional staircase mechanism to Algorithm~\ref{alg:new} with noise from the \algname mechanism in Figure \ref{fig:low_privacy_eps_vary}, focusing on the low-privacy regime where $\epsilon$ is large. The lower bound in Figure~\ref{fig:low_privacy_eps_vary} is again $\sigma^2(\epsilon)$, but notice now that as $\epsilon$ grows, the \algname mechanism continues to match the lower-bound indefinitely, and increasingly outperforms the two-dimensional staircase mechanism. In Figure~\ref{fig:low_privacy_mu_vary}  (Appendix~\ref{app:experiments}), we show that both mechanisms also have a similar M-shaped error curve over $\mu$. In Appendix~\ref{app:experiments}, we provide additional experiments to demonstrate that the proposed algorithm outperforms existing algorithms on small datasets.

\conf{\section{Impact statement}
}
\arxiv{
\section{
 Acknowledgements}
The authors thank Travis Dick, Gautam Kamath, and Ziteng Sun for helpful comments and suggestions.}
\bibliographystyle{abbrvnat}
\bibliography{references}
\newpage
\onecolumn
\appendix

\section{Properties of staircase mechanisms}
We first define both the one-dimensional and two-dimensional staircase mechanisms.
\begin{definition}[Univariate staircase mechanism]
\label{def:1d_stair}
Let $\Delta$ be the sensitivity of the underlying query. The univariate staircase mechanism \cite{geng2014optimal} is parameterized by $\gamma \in [0, 1]$ and is given by
\begin{equation}
f_\gamma(x) = 
    \begin{cases}
     a(\gamma) , & x \in [0, \gamma \Delta)\\
     a(\gamma) e^{-\epsilon} , & x \in [\gamma \Delta,\Delta)\\
     e^{-k \epsilon} f_\gamma(x - k \Delta), & x \in [k \Delta,(k+1)\Delta) \\
     f_\gamma(-x), & x < 0.
    \end{cases}
\end{equation}
Here $a(\gamma)$ is the normalization factor and given by
\[
a(\gamma) = \frac{1-e^{-\epsilon}}{2 \Delta \left(\gamma + e^{-\epsilon}(1 - \gamma) \right)}.
\]
\end{definition}
\begin{definition}[Two-dimensional staircase mechanism]
\label{def:2d_stair}
Let $\Delta$ be the sensitivity of the underlying query.  The two-dimensional staircase mechanism \cite{geng2015staircase} is parameterized by $\gamma \in [0, 1]$ and is given by
\begin{equation}
f_\gamma(x,y) = 
    \begin{cases}
     a(\gamma) e^{-\epsilon k}, & |x| + |y| \in [k\Delta, (k + \gamma)\Delta)\\
    a(\gamma) e^{-\epsilon (k+1)}, & |x| + |y| \in[(k+\gamma)\Delta + k+1)\Delta),
    \end{cases}
\end{equation}
Here $a(\gamma)$ is the normalization factor and given by
\[
a(\gamma) = \frac{(1-e^{-\epsilon})^2}{2 \left( e^{-\epsilon}  (e^{-\epsilon} + (1-e^{-\epsilon})\gamma) + (1-e^{-\epsilon}) (e^{-\epsilon} + (1-e^{-\epsilon})\gamma^2) \right)}.
\]

\end{definition}
We prove the following result for the MSE of the two-dimensional staircase mechanism.
\begin{lemma}
\label{lem:two_stair}
Let $\epsilon \geq \epsilon_0$ for a sufficiently large value of $\epsilon_0$. Then, for the two-dimensional stair-case mechanism,
\[
\min_{\gamma} \EE_{f_{\gamma}}[x^2] = 
\min_{\gamma} \EE_{f_{\gamma}}[y^2] = \Delta^2 \cdot \Theta \left( e^{-\epsilon/2} \right),
\]
\end{lemma}
\begin{proof}
Without loss of generality, we assume the sensitivity to be one. By symmetry,
\[
\min_{\gamma} \EE_{f_{\gamma}}[x^2] = 
\min_{\gamma} \EE_{f_{\gamma}}[y^2],
\]
and it suffices to consider $ \EE_{f_{\gamma}}[x^2]$. For $k \geq 1$, it can be shown that
\[
\int_{|x| + |y| \in [k, k+1)} x^2 dx dy = \Theta \left( k^3 \right).
\]
Similarly,
\[
\int_{|x| + |y| \in [\gamma, 1)} x^2 dx dy = \Theta \left(1\right),
\]
and
\[
\int_{|x| + |y| \in [0, \gamma)} x^2 dx dy = \Theta \left(\gamma^4\right).
\]
Furthermore,
\[
a(\gamma)  =\Theta \left(\frac{ 1}{e^{-\epsilon} + \gamma^2}
\right).
\]
Combining the above four equations yields,
\begin{align*}
\EE_{f_{\gamma}}[x^2]
& = a(\gamma) \Theta \left( \gamma^4 + e^{-\epsilon} + \sum_{k \geq 1} ( e^{-k\epsilon} + e^{(k+1)\epsilon}) k^3 \right) \\
& = a(\gamma) \Theta \left( \gamma^4 + e^{-\epsilon} \right) \\
& =  \Theta \left(\frac{ \gamma^4 + e^{-\epsilon}}{e^{-\epsilon} + \gamma^2}
\right).
\end{align*}
Minimizing over $\gamma$ yields the desired result.
\end{proof}
We next prove a result on the moments of both one-dimensional and two-dimensional staircase mechanisms.
\begin{lemma}
\label{lem:stair_moments}
Let $16 \geq m \geq 0$. For both the one-dimensional and two-dimensional staircase mechanisms that guarantees $\epsilon$-differential privacy for sensitivity one queries,
\[
\EE[x^m] = \frac{1}{\epsilon^{m-2}} \cdot O \left(\EE[x^2]\right)
\]
\end{lemma}
\begin{proof}
We provide the proof for the univariate staircase mechanism. The proof for the two-dimensional stair case mechanism is similar and omitted. Let $a(\gamma)$ denote the normalization term in the definition of univariate staircase mechanism \citet[equation 26]{geng2015staircase}. Since $\gamma \leq 1$,
\begin{align*}
\EE_{f_{\gamma}}[x^m]
& = a(\gamma) \Theta \left( \gamma^m + \sum_{k \geq 1} e^{-k \epsilon} k^m \right) \\
& \stackrel{(a)}{=} a(\gamma) O \left( \gamma^2 + \sum_{k \geq 1} e^{-k \epsilon} k^m \right) \\
& \stackrel{(b)}{=} a(\gamma) O \left( \gamma^2 + \frac{1}{(1-e^{-\epsilon})^{m-2}} \sum_{k \geq 1} e^{-k \epsilon} k^2 \right) \\
& =\frac{1}{(1-e^{-\epsilon})^{m-2}} a(\gamma) O \left( \gamma^2 + \sum_{k \geq 1} e^{-k \epsilon} k^2 \right) \\
& = \frac{1}{(1-e^{-\epsilon})^{m-2}} \cdot O \left(\EE_{f_{\gamma}}[x^2] \right) \\
& \stackrel{(c)}{=} \frac{1}{\epsilon^{m-2}} \cdot O \left(\EE_{f_{\gamma}}[x^2] \right),
\end{align*}
where $(a)$ uses the fact that $\gamma \leq 1$, $(c)$ follows from the fact that $1-e^{-\epsilon} \leq \epsilon$. To prove $(b)$ notice that 
\[
(1-e^{-\epsilon}) \sum_{k \geq 1} e^{-k\cdot \epsilon} k^{m} = \sum_{k \geq 1} e^{-k\cdot \epsilon} (k^m - (k-1)^m) = O \left( \sum_{k \geq 1} e^{-k\cdot \epsilon} k^{m-1} \right).
\]
\end{proof}
We next state a concentration property of both one-dimensional and two-dimensional staircase mechanisms.
\begin{lemma}
\label{lem:stair_concentration}
Let $n \geq 10$. For both the one-dimensional and two-dimensional staircase mechanisms that guarantees $\epsilon$-differential privacy for sensitivity one,
\[
\Pr(x \leq -n) \leq e^{-n\epsilon / 2}.
\]
\end{lemma}
\begin{proof}
We provide the proof for the univariate staircase mechanism. The proof for the two-dimensional stair case mechanism is similar and omitted. Let $a(\gamma)$ denote the normalization term in the definition of univariate staircase mechanism \citet[equation 26]{geng2015staircase}. Note that for the univariate staircase mechanism,
\[
a(\gamma) \leq \frac{1-e^{-\epsilon}}{2 e^{-\epsilon}}.
\]
We now bound the desired quantity as follows.
\begin{align*}
\Pr(x \leq -n) & \leq \sum_{k\geq n} a(\gamma) e^{-k \epsilon}  \\
 & = a(\gamma) \frac{e^{-n \epsilon}}{1-e^{-\epsilon}}   \\
 & \leq (1-e^{-\epsilon}) e^{\epsilon} \frac{e^{-n \epsilon}}{1-e^{-\epsilon}}   \\
 & = e^{-(n-1)\epsilon},
\end{align*}
where the last inequality follows from the bound on $a(\gamma)$.
\end{proof}

\section{Proof of Lemma~\ref{lem:technical}}
\label{app:technical}
We first focus on the upper bound.
\begin{align*}
 \left(\clip\left( \frac{a + Z_a}{b+Z_b}  \right)- \frac{a}{b}   \right)^2
 & = \left(\clip\left( \frac{a + Z_a}{b+Z_b}  \right)- \frac{a}{b}   \right)^2 1_{Z_b \geq -b/2} +  \left(\clip\left( \frac{a + Z_a}{b+Z_b}  \right)- \frac{a}{b}   \right)^2 1_{Z_b < -b/2} \\
  & \leq  \left(\clip\left( \frac{a + Z_a}{b+Z_b}  \right)- \frac{a}{b}   \right)^2 1_{Z_b \geq -b/2} +  4M^2 1_{Z_b < -b/2} \\
   & \leq  \left(\clip\left( \frac{a + Z_a}{b+Z_b}  \right)- \frac{a}{b}   \right)^2 1_{Z_b \geq -b/2} + 4M^2 1_{Z_b < -b/2},
\end{align*}
where the first inequality uses the fact that both $\clip\left( \frac{a + Z_a}{b+Z_b}  \right)$ and $\frac{a}{b}$ lie in $[0,1]$ and the last inequality uses the fact that clipping is a projection operator. Taking expectation on both sides yield,
\begin{align*}
\EE \left[  \left(\clip\left( \frac{a + Z_a}{b+Z_b}  \right)- \frac{a}{b}   \right)^2 \right]
& \leq \EE \left[\left( \frac{a + Z_a}{b+Z_b} - \frac{a}{b}  \right)^2 1_{Z_b  \geq -b/2}\right] + 4M^2\Pr(Z_b < -b/2).
\end{align*}
We now use algebraic manipulation to simplify $ \frac{a + Z_a}{b+Z_b} - \frac{a}{b}$.
\begin{align}
 \frac{a + Z_a}{b+Z_b} - \frac{a}{b}
 & = \frac{Z_a}{b} + \frac{a + Z_a}{b+Z_b} - \frac{a+Z_a}{b} \nonumber \\
& =  \frac{Z_a}{b} - \frac{(a + Z_a)(Z_b)}{(b+Z_b)(b)}\nonumber \\
& =  \frac{Z_a}{b} - \frac{aZ_b}{(b+Z_b)(b)} -  \frac{Z_a Z_b}{(b+Z_b)(b)} \nonumber\\
& =  \frac{Z_a}{b} - \frac{aZ_b}{b^2} +  \frac{aZ_b}{b^2}  -  \frac{aZ_b}{(b+Z_b)(b)} -  \frac{Z_a Z_b}{(b+Z_b)(b)} \nonumber\\
& =  \frac{Z_a}{b} - \frac{aZ_b}{b^2} +  \frac{aZ^2_b}{b^2(b+Z_b)} -  \frac{Z_a Z_b}{(b+Z_b)(b)}.\nonumber 
\end{align}
Let $C = \frac{Z_a}{b} - \frac{aZ_b}{b^2}$ and $D = \frac{aZ^2_b}{b^2(b+Z_b)} -  \frac{Z_a Z_b}{(b+Z_b)(b)}$.  
\begin{align*}
\EE \left[\left( \frac{a + Z_a}{b+Z_b} - \frac{a}{b}  \right)^2 1_{Z_b  \geq - b/2}\right]
& = \EE \left[\left(C+D \right)^2 1_{Z_b  \geq -b/2}\right] \\
& = \EE \left[C^2 1_{Z_b  \geq -b/2}\right] + \EE \left[D^2 1_{Z_b  \geq- b/2}\right] + \EE \left[2 C D 1_{Z_b  \geq -b/2}\right] \\
& \leq \EE \left[C^2 1_{Z_b  \geq -b/2}\right] + \EE \left[D^2 1_{Z_b  \geq -b/2}\right] + 2 \sqrt{\EE \left[ C^2 \right] \EE\left[ D^2 1_{Z_b  \geq -b/2}\right]},
\end{align*}
where the last inequality uses Cauchy-Schwarz inequality. We next upper bound $D$. If $Z_b \geq -b/2$, and $|a| \leq b M$, then
\begin{align*}
|D| =\left \lvert \frac{aZ^2_b}{b^2(b+Z_b)} -  \frac{Z_a Z_b}{(b+Z_b)(b)} \right \rvert
& \leq \left \lvert \frac{aZ^2_b}{b^2(b+Z_b)}\right \rvert + \left \lvert  \frac{Z_a Z_b}{(b+Z_b)(b)} \right \rvert  \\
& \leq  \left \lvert \frac{MZ^2_b}{b(b+Z_b)}\right \rvert + \left \lvert  \frac{Z_a Z_b}{(b+Z_b)(b)} \right \rvert  \\
& \leq \left \lvert \frac{2MZ^2_b}{b^2}\right \rvert + \left \lvert  \frac{2Z_a Z_b}{b^2} \right \rvert.
\end{align*}
Let $F = \left \lvert \frac{2MZ^2_b}{b^2}\right \rvert + \left \lvert  \frac{2Z_a Z_b}{b^2} \right \rvert $. Combing the above bound with previous equations yields the upper bound:
\begin{align*}
\EE \left[\left( \frac{a + Z_a}{b+Z_b} - \frac{a}{b}  \right)^2 1_{Z_b  \geq -b/2}\right]
& \leq  \EE \left[C^2 1_{Z_b  \geq -b/2}\right] + \EE \left[D^2 1_{Z_b  \geq -b/2}\right] + 2 \sqrt{\EE \left[ C^2 \right] \EE\left[ D^2 1_{Z_b  \geq -b/2}\right]} \\
& \leq   \EE \left[C^2 1_{Z_b  \geq -b/2}\right] + \EE \left[F^2 1_{Z_b  \geq -b/2}\right] + 2 \sqrt{\EE \left[ C^2 \right] \EE\left[ F^2 1_{Z_b  \geq -b/2}\right]} \\
& \leq   \EE \left[C^2\right] + \EE \left[F^2\right] + 2 \sqrt{\EE \left[ C^2 \right] \EE\left[ F^2\right]}.
\end{align*}
\section{Proof of Lemma~\ref{lem:current}}
\label{app:current}
The differential privacy guarantee follows by the properties of Laplace mechanism and composition theorem and in the rest of the proof, we focus on the MSE guarantees. The analysis of MSE heavily relies in Lemma~\ref{lem:technical}. Let $n = |D|$, $a = s$, $b = n$, $Z_a = Z_s \sim \lap(w/\varepsilon)$, $Z_b = Z_n = \lap(2/\varepsilon)$, and $M = (u - \ell)/2$. With these definitions, to apply Lemma~\ref{lem:technical}, we need to bound $\EE[C^2]$, $\EE[F^2]$, and $\Pr(Z_b < -b/2)$. We bound each of the terms below.
\begin{align*}
\EE[C^2] & = \EE\left[\left(\frac{Z_s}{n} - \frac{(\mu-m) Z_n}{n} \right)^2 \right] \\ 
 &= \frac{2(u-l)^2}{n^2\varepsilon^2} + \frac{8(\mu-m)^2}{n^2\varepsilon^2}.
\end{align*}
Since $(x+y)^2 \leq 2 x^2 + 2 y^2$, 
\begin{align*}
\EE[F^2] &  \leq 8\frac{(u-\ell)^2}{n^4} \EE \left[Z^4_n\right] + 8\frac{1}{n^4} \EE \left[Z^2_a  Z^2_b\right]  \\
& = o \left(\frac{(u-l)^2}{n^2\varepsilon^2} \right).
\end{align*}
Finally, by the tail bounds of the Laplace mechanism,
\begin{align*}
M^2 \Pr(Z_b < -b/2)  & \leq (u - \ell)^2 \Pr(Z_n < -n/2) \\
& = o \left(\frac{(u-l)^2}{n^2\varepsilon^2} \right).
\end{align*}
Combining the above three equations together with Lemma~\ref{lem:technical} yields the lemma.

\section{Analysis of the \capalgname Mechanism}
\subsection{Computing normalization constant}
\label{app:norm}
Let $c(\gamma)$ be the normalizing factor that ensures the sum of the probabilities is one.

From the paper \citet{geng2014optimal}, we know the normalizing factor for the univariate staircase mechanism to be 
$$a(\gamma) \triangleq \frac{1-e^{-\epsilon}}{2(\gamma + ^{-\epsilon}(1-\gamma))}.$$
From Lemma \ref{lem:marginal_1d_staircase}, the marginal of the \algname distribution is the univariate staircase distribution. hence,
\begin{align*}
\sum_{y} f_\gamma(x, y) = a(\gamma),
\end{align*}
where $a(\gamma)$ is the normalization factor in the univariate staircase distribution. Observe that if $x =0$, then for all integer $y$, $(x,y) \in A_y$ and hence
\[
\sum_{y} f_\gamma(x, y) = \sum_{y} c(\gamma) e^{-\epsilon y} = c(\gamma) \frac{1-e^{-\epsilon}}{1+e^{-\epsilon}}.
\]
Hence,
\[
c(\gamma) = \frac{1-e^{-\epsilon}}{1+e^{-\epsilon}}a(\gamma) = \frac{(1-e^{-\epsilon})^2}{2\Delta(1+e^{-\epsilon})(\gamma + e^{-\epsilon}(1-\gamma))}.
\]
\subsection{Proof of Theorem \ref{thm:hourglass_dp}}
\label{app:hour_privacy}

Let $D$ and $D'$ be two neighboring datasets such that $q(D) - Q(D_0) = -(x_0, 1-x_0)$. Let let $\tilde{q}(D)$ denote the output of the \algname mechanism. To provide differential privacy guarantee, it suffices to prove 
upper and lower bounds for 
\[
 \frac{\Pr(\tilde{q}(D) = (x, y))}{\Pr(\tilde{q}(D') = (x,y))}.
\]
Let $\bar{Z}$ be a sample from the \algname mechanism, then
\begin{align*}
 \frac{\Pr(\tilde{q}(D) = (x, y))}{\Pr(\tilde{q}(D') = (x,y))} 
 & = 
 \frac{\Pr({q}(D) + \bar{Z} = (x, y))}{\Pr({q}(D') + \bar{Z} = (x,y))} \\
 & = 
 \frac{\Pr({q}(D') - (x_0, 1-x_0) + \bar{Z} = (x, y))}{\Pr({q}(D') + \bar{Z} = (x,y))} \\
  & = 
 \frac{\Pr( \bar{Z} = (x, y) - {q}(D')+ (x_0, 1-x_0))}{\Pr(\bar{Z} = (x, y) - {q}(D'))} \\
 & = 
 \frac{\Pr( \bar{Z} = (x', y') + (x_0, 1-x_0))}{\Pr(\bar{Z} = (x', y')}  \\
  & = 
 \frac{f_{\gamma}(x' + x_0, y' + 1 - x_0) )}{f_{\gamma}(x', y')},
\end{align*}
where $(x', y') = (x, y) - {q}(D')$. Hence, to prove the privacy guarantee, it suffices to prove upper and lower bounds on the ratio, 
\[
R(x, y, x_0) \triangleq \frac{f_{\gamma}(x + x_0, y + 1 - x_0) )}{f_{\gamma}(x, y)},
\]
for all $x, y$ and $x_0 \in [0, 1]$. Without loss of generality, we assume that $\Delta = 1$.  Let $\mathcal{S} = \{A_k, B_k(i): \forall k, i \}$. Note that $\mathcal{S}$ partitions the domain of $f_\gamma$ into disjoint partitions. We observe that if $(x, y) \in A_{-1}$, then $f_{\gamma}(x, y) = c(\gamma) e^{-\epsilon}$ and for all $x_0$, $f_{\gamma}(x + x_0, y + 1 - x_0) \in \{c(\gamma) , c(\gamma)e^{-2\epsilon} \}$, hence if $(x, y) \in A_{-1}$ then $R(x, y, x_0)  \in \{e^{-\epsilon}, e^{\epsilon} \}$ for all $x_0$. If $(x, y) \notin A_{-1}$, then $(x+x_0, y+1-x_0)$ belongs to at most two sets in $\mathcal{S}$ and hence  $f_{\gamma}(x + x_0, y + 1 - x_0)$ is monotonic in $x_0$ and proving the result for $x_0 \in \{0, 1\}$  suffices.

We now focus on the case when $x_0 \in \{0, 1\}$. Of these two cases, by symmetry it suffices to consider $x_0 = 0$. Furthermore, we show the result when $x +y = k$ for some $x \geq 0$. The proof for the other side is similar and omitted. We prove the result by dividing the problem into subcases depending the value of $(x, y)$. Subcase $(A)$: If $(x, y) \in A_k$ for some $k$, then $(x, y+1) \in A_{k+1}$ and hence $R(x, y, 0) = e^{-\epsilon}$.  Subcase $(B1)$: If $(x, y) \in B_k(1)$, then $(x, y) \in A_{k+1}$ and hence $R(x, y, 0) = e^{\epsilon}$. Subcase $(B2)$: If $(x, y) \in B_k(i)$ for $i \geq 2$, then $(x, y) \in B_{k+1}(i-1)$ and hence $R(x, y, 0) = e^{\epsilon}$. Hence, we have shown that for all $x,y$ and $x_0 \in \{0, 1\}$,
\[
R(x, y, x_0)  \in \{e^{-\epsilon}, e^{\epsilon}\}.
\]

\section{Proof of Lemma~\ref{lem:two_stair_mean}}
\label{app:two_stair_mean}
The privacy guarantee is similar to that of Theorem~\ref{thm:new} and is omitted.
As before, the proof of utility heavily relies in Lemma~\ref{lem:technical}. 
Let $n = |D|$. Observe that 
\[
\EE[(\hat{\mu} - \mu)^2] \leq (u - \ell)^2 \EE \left[ \left(\frac{ \hat s_1}{\hat s_1 + \hat s_2} - \frac{s_1}{n}  \right)^2\right].
\]
Let $a = s_1$, $b = n$, $Z_a = Z_{1}$, $Z_b = Z_1 + Z_2$, where $Z_1, Z_2$ are from the two-dimensional staircase mechanism. Let $M = 1$. With these definitions, to apply Lemma~\ref{lem:technical}, we need to bound $\EE[C^2]$, $\EE[F^2]$, and $\Pr(Z_b < -b/2)$. We bound each of the terms below.
Let $\alpha = \frac{s_1}{n}$.
\begin{align*}
& \EE[C^2] \\
& = \EE\left[\left(\frac{Z_1}{n} - \frac{(\frac{s_1}{n}) (Z_1 + Z_2)}{n} \right)^2 \right] \\
& = \frac{1}{n^2} \EE \left[(1-\alpha)^2 Z^2_1 + \alpha^2 Z^2_2 - 2 (1-\alpha )\alpha Z_1, Z_2 \right] \\
& \leq \frac{1}{n^2} \EE \left[(1-\alpha)^2 Z^2_1 + \alpha^2 Z^2_2\right] + 2 (1-\alpha )\alpha \sqrt{\EE\left[ Z^2_1, Z^2_2 \right]} \\
& \leq \frac{1}{n^2} (1-\alpha)^2\widetilde{\sigma}^2(\epsilon) + \alpha^2 \widetilde{\sigma}^2(\epsilon) + 2 (1-\alpha )\alpha \widetilde{\sigma}^2(\epsilon) \\
& = \frac{\widetilde{\sigma}^2(\epsilon)}{n^2},
\end{align*}
where the first inequality follows by the Cauchy-Schwarz inequality. Since $(x+y)^2 \leq 2 x^2 + 2 y^2$,
\begin{align*}
\EE[F^2]  & \leq \frac{8}{n^4} \EE[(Z_1 + Z_2)^4] + \frac{8}{n^4} \EE[(Z_1 + Z_2)^2 Z^2_1] \\
& = o \left( \frac{\widetilde{\sigma}^2(\epsilon)}{n^2} \right),
\end{align*}
where the last equality follows from Lemma~\ref{lem:stair_moments}.
Finally,
\begin{align*}
\Pr(Z_b < -b/2)  & \leq \Pr(Z_1 + Z_2 < -n/2) \\
& \leq \Pr(Z_1  < -n/4) + \Pr( Z_2 < -n/4) \\ 
&= o \left( \frac{\widetilde{\sigma}^2(\epsilon)}{n^2} \right),
\end{align*}
where the last equality follows from Lemma~\ref{lem:stair_concentration}.
Combining the above three equations together with Lemma~\ref{lem:technical} and observing the fact that $\mu = \ell + \frac{s_1w}{n}$ yields the result.

\section{Proof of Theorem~\ref{thm:hourglass_main}}
\label{app:hourglass_main}
The privacy guarantee is similar to that of Theorem~\ref{thm:new} and is omitted.
As before, the proof of utility heavily relies in Lemma~\ref{lem:technical}. 
Let $n = |D|$. Observe that 
\[
\EE[(\hat{\mu} - \mu)^2]  \leq (u - \ell)^2 \EE \left[ \left(\frac{ \hat s_1}{\hat s_1 + \hat s_2} - \frac{s_1}{n}  \right)^2\right].
\]
Let $a = s_1$, $b = n$, $Z_a = Z_{1}$, $Z_b = Z_1 + Z_2$, where $Z_1, Z_2$ are from the \algname  mechanism. Let $M = 1$. With these definitions, to apply Lemma~\ref{lem:technical}, we need to bound $\EE[C^2]$, $\EE[F^2]$, and $\Pr(Z_b < -b/2)$. We bound each of the terms below.
Let $\alpha = \frac{s_1}{n}$.
\begin{align*}
& \EE[C^2] \\
& = \EE\left[\left(\frac{Z_1}{n} - \frac{(\frac{s_1}{n}) (Z_1 + Z_2)}{n} \right)^2 \right] \\
& = \frac{1}{n^2} \EE \left[(1-\alpha)^2 Z^2_1 + \alpha^2 Z^2_2 - 2 (1-\alpha )\alpha Z_1, Z_2 \right] \\
& \leq \frac{1}{n^2} \EE \left[(1-\alpha)^2 Z^2_1 + \alpha^2 Z^2_2\right] + 2 (1-\alpha )\alpha \sqrt{\EE\left[ Z^2_1, Z^2_2 \right]} \\
& \leq \frac{1}{n^2} (1-\alpha)^2\widetilde{\sigma}^2(\epsilon) + \alpha^2 \widetilde{\sigma}^2(\epsilon) + 2 (1-\alpha )\alpha \widetilde{\sigma}^2(\epsilon) \\
& = \frac{{\sigma}^2(\epsilon)}{n^2},
\end{align*}
where the first inequality follows by the Cauchy-Schwarz inequality. Since $(x+y)^2 \leq 2 x^2 + 2 y^2$,
\begin{align*}
\EE[F^2]  & \leq \frac{8}{n^4} \EE[(Z_1 + Z_2)^4] + \frac{8}{n^4} \EE[(Z_1 + Z_2)^2 Z^2_1] \\
& = o \left( \frac{{\sigma}^2(\epsilon)}{n^2} \right),
\end{align*}
where the last equality follows from Lemma~\ref{lem:stair_moments} and the fact that the marginal distribution of \algname mechanism is same as the staircase mechanism..
Finally,
\begin{align*}
\Pr(Z_b < -b/2)  & \leq \Pr(Z_1 + Z_2 < -n/2) \\
& \leq \Pr(Z_1  < -n/4) + \Pr( Z_2 < -n/4) \\ 
&= o \left( \frac{{\sigma}^2(\epsilon)}{n^2} \right),
\end{align*}
where the last equality follows from Lemma~\ref{lem:stair_concentration} and the fact that the marginal distribution of \algname mechanism is same as the staircase mechanism.
Combining the above three equations together with Lemma~\ref{lem:technical} and observing the fact that $\mu = \ell + \frac{s_1 w}{n}$ yields the theorem.

\section{Proof of Theorem~\ref{thm:lower}}
\label{app:lower}

We first state the following lemma, which removes the dependence on $\ell$ and $u$. 
\begin{lemma}
\label{lem:zeroone}
\[
R_{\ar}(\varepsilon, \ell, u) = (u - \ell)^2 R_{\ar}(\varepsilon, 0, 1).
\]
\end{lemma}
\begin{proof} 
Given a dataset from $D \in \cD^*(\ell, u)$, one can create a dataset in $f(D) \in \cD^*(0, 1)$ by applying the transformation $f(x) = \frac{x - \ell}{u-\ell}$ to each of the points. Let $\hat{\mu}$ be a mean estimation algorithm for datasets in $\cD^*(0, 1)$, then given a dataset from $D \in \cD^*(\ell, u)$, one can scale all points by applying $f$ and compute the output as $\hat{\mu}'(D) = f^{-1}(\hat{\mu}(f(D))$. If $\hat{\mu}$ is an $\varepsilon$-differentially private algorithm, then $\hat{\mu}'$ is also an $\varepsilon$-differentially private algorithm. Furthermore, the utilities are related by 
\[
L(\hat{\mu}', D) \leq (u-\ell)^2 L(\hat{\mu}, f(D)).
\]
Taking supremum over datasets $D$ and infimum over all differentially private algorithms yields
\[
R_{\ar}(\varepsilon, \ell, u) \leq (u - \ell)^2 R_{\ar}(\varepsilon, 0, 1).
\]
The proof for the other direction is similar and omitted.
\end{proof}

\begin{proof}[Proof of Theorem~\ref{thm:lower}]
By Lemma~\ref{lem:zeroone}, it suffices to consider the scenario when $\ell = 0$ and $u = 1$.
Let $\cD^k$ and $S(D)$ be the same as those be defined as in Lemma~\ref{lem:conjecture}. Let $\widetilde{\cD}^k$ be the set of all datasets obtaining by combining each dataset in $\cD^k$ with $n$ values of ones. Observe for any dataset $D$ in $\widetilde{\cD}^k$, 
\[
\mu(D) = \frac{S(D)}{n+k} \leq \frac{k}{n}.
\]
Suppose we have an $\varepsilon$-differentially private estimator $\hat{\mu}$ on $\cD^k$. We convert it to an estimator of $S(D)$ as
\[
\hat{S} = n \hat{\mu}(D).
\]
For this estimator, based on the Lemma~\ref{lem:conjecture}, there exists a $D \in \cD^k$ such that 
\[
 \EE \left[(\hat{S} - S(D))^2\right] \geq \frac{2}{\varepsilon^2} ( 1 - o(1))
\]
and hence
\[
\EE \left[n \hat{\mu} - S(D))^2\right] \geq \frac{2}{\varepsilon^2} ( 1 - o(1)),
\]
which implies
\[
|D|^2 \EE \left[n^2 \hat{\mu} - S(D))^2\right] \geq \frac{2}{\varepsilon^2} ( 1 - o(1)).
\]
We now upper bound the left hand side of the above expression.
\begin{align*}
& \EE \left[(n \hat{\mu} - S(D))^2\right] \\ 
& = \EE \left[(n \hat{\mu} - n \mu  + n\mu - S(D))^2\right] \\
& = \EE \left[(n \hat{\mu} - n \mu  + n\mu - S(D))^2\right] \\
& = \EE \left[(n \hat{\mu} - n \mu  - \mu(D) k)^2\right] \\
& = \EE \left[(n \hat{\mu} - n \mu)^2\right] + \EE \left[( \mu(D) k)^2\right] \\
& - 2 \EE \left[(n \hat{\mu} - n \mu)( \mu(D) k)\right] \\
& \leq \EE \left[(n \hat{\mu} - n \mu)^2\right] + \frac{k^4}{n^2} + 2 \sqrt{\EE \left[(n \hat{\mu} - n \mu)^2\right]} \frac{k^2}{n} \\
& \leq \EE \left[(n \hat{\mu} - n \mu)^2\right] + \frac{k^4}{n^2} + \frac{k^3}{n} \\
& \leq |D|^2 \EE \left[(\hat{\mu} -  \mu)^2\right] + \frac{k^4}{n^2} + \frac{k^3}{n},
\end{align*}
where the first inequality follows by Cauchy-Schwarz inequality and the second inequality follows by observing that both $\hat{\mu}$ and $\mu$ lie in $[0, k/n]$. Setting $k = o(n^{1/3})$ yields the following theorem.
\end{proof}

\section{Additional experiments}
\label{app:experiments}
\begin{figure}[t]
        \centering
        \includegraphics[scale=0.6]{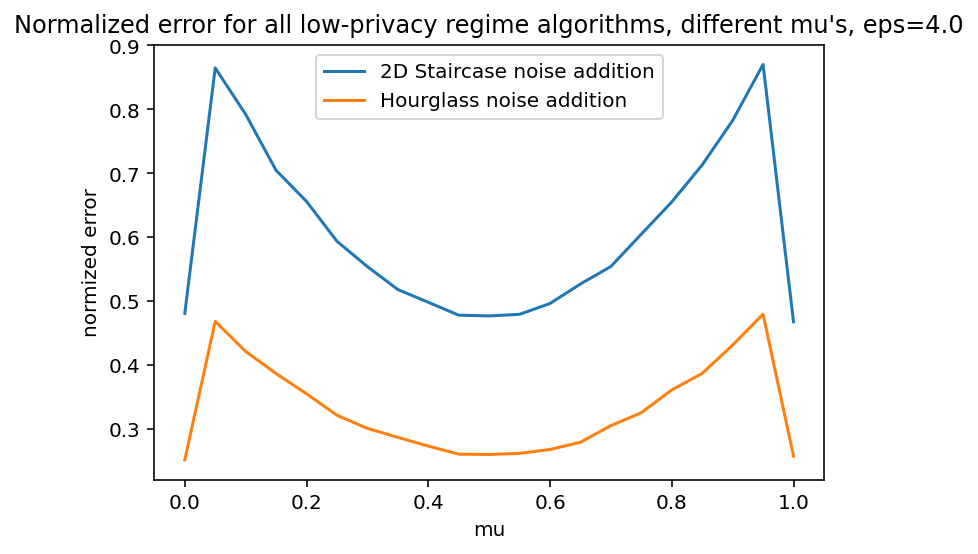}
        \caption{\algname and two dimensional staircase mechanisms, low privacy regime, varying $\mu$.}
           \label{fig:low_privacy_mu_vary}
\end{figure}

In this section, we explore how the MSEs of different algorithms change when the the size of the dataset $|D|$ changes.  The errors are normalized by $\frac{ \epsilon^2}{2}$ to keep them in a similar range across different privacy regimes. We choose $|D| \in \{2^i\}_{i=0, \ldots, 6}$ and observe the trend of changes in MSEs for all listed algorithms.

\begin{figure*}[h]
    \centering
    \begin{subfigure}[t]{0.4\textwidth}
    \includegraphics[scale=0.4]{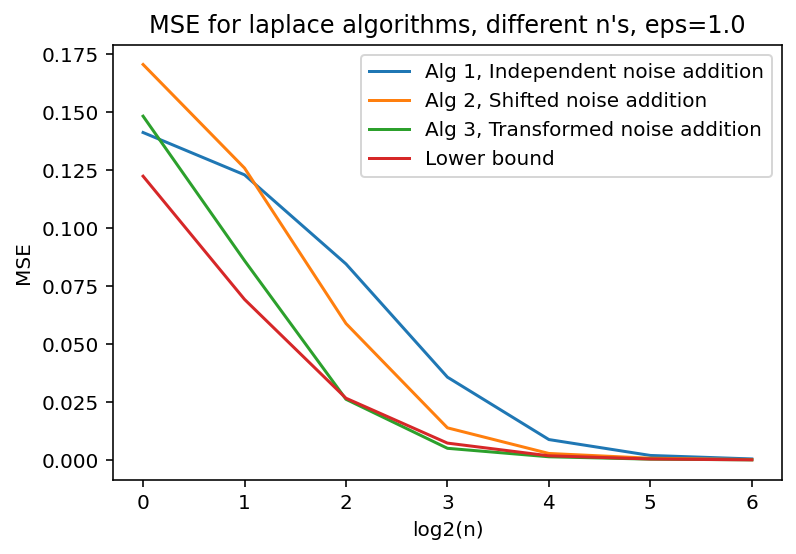}
    \subcaption{All Laplace mechanisms, high privacy regime, varying $|D|$.}
    \label{fig:high_privacy_all_ns}
    \end{subfigure}
    ~~
    \begin{subfigure}[t]{0.4\textwidth}
    \includegraphics[scale=0.4]{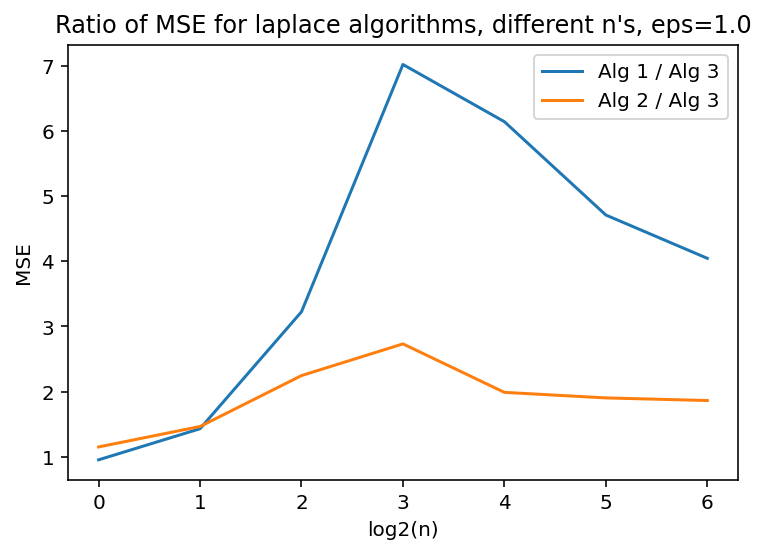}
    \subcaption{Ratio of MSEs, high privacy varying $|D|$, other Laplace mechanisms v.s. Transformed.}
    \label{fig:high_privacy_all_ns_ratio}
    \end{subfigure}
    ~~
        \begin{subfigure}[t]{0.4\textwidth}
    \includegraphics[scale=0.4]{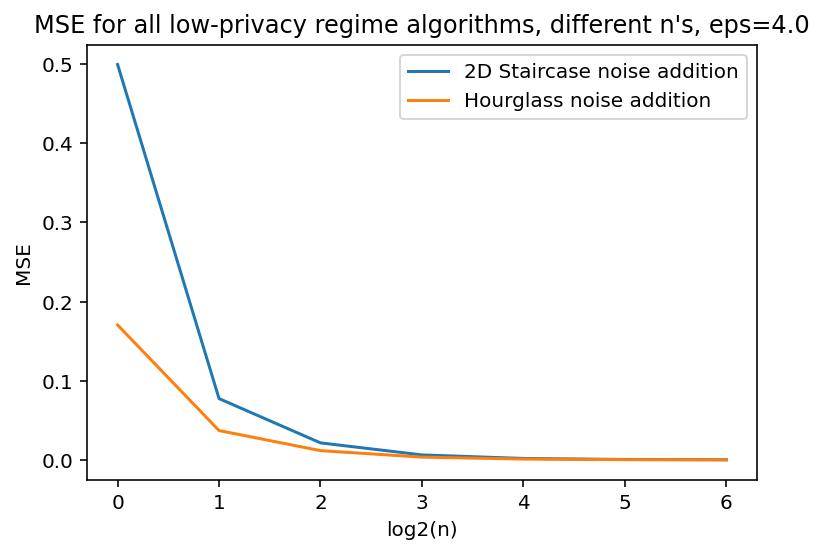}
    \subcaption{Ratio of MSEs, low privacy varying $|D|$, 2D staircase and \algname mechanisms.}
    \label{fig:low_privacy_all_ns_discrete}
    \end{subfigure}
    ~~
           \begin{subfigure}[t]{0.4\textwidth}
    \includegraphics[scale=0.4]{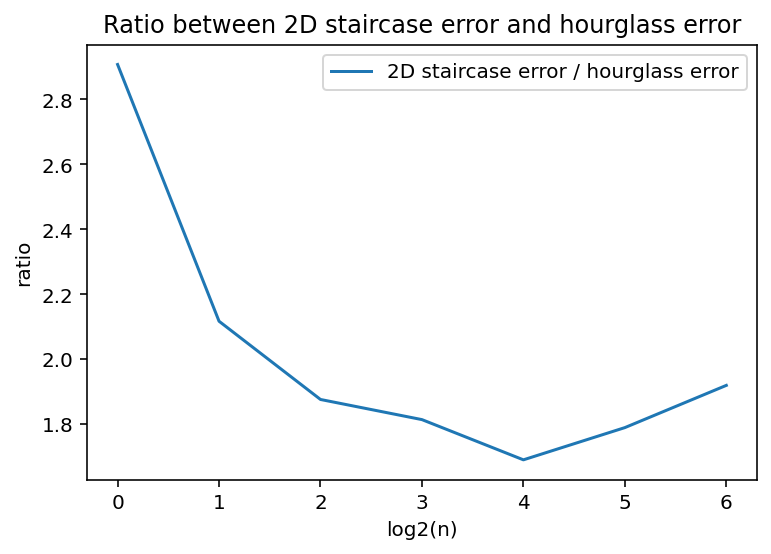}
    \subcaption{Ratio of MSEs, low privacy varying $|D|$, 2D-staircase v.s. \algname.}
    \label{fig:low_privacy_all_ns_ratio}
    \end{subfigure}
    \label{fig:all_ns}
\end{figure*}
Figure \ref{fig:high_privacy_all_ns} shows how the MSEs (normalized by $\frac{\epsilon^2}{2}$) of different Laplace mechanisms vary as a function of $|D|$  for $\epsilon = 1.0$ and $\mu = 0.01$. The lower bound (red curve), as in Section~\ref{sec:experiments}, is developed from using the univariate staircase mechanism with optimal $\gamma$ on private mean in the swap model. One can see that only Algorithm~\ref{alg:new} converges to the lower bound as $|D|$ grows bigger, being about two times better than Algorithm~\ref{alg:current} when $|D| \geq 16$.

On the other hand, Figure  \ref{fig:low_privacy_all_ns_discrete} and \ref{fig:low_privacy_all_ns_ratio} compare the  MSEs (normalized by $\frac{\epsilon^2}{2}$) of the two-dimensional staircase and the \algname mechanism in the low privacy regime when $\epsilon = 4.0$ and $\mu = 0.01$. The ratio of the former over latter is found to be always bigger than one, showing \algname mechanism to be better.

\end{document}